\documentclass{article}
\pdfoutput=1

\usepackage{amsthm}
\usepackage{amsfonts}
\usepackage{pifont}
\usepackage{mathpazo}

\usepackage{xspace,enumerate}
\usepackage{amsmath,amssymb}
\usepackage[toc,page]{appendix}
\usepackage{thmtools}
\usepackage{thm-restate}
\usepackage{graphicx}
\usepackage{csquotes}
\usepackage{makecell}
\usepackage{tabularx}
\usepackage{relsize}
\usepackage{scalerel}
\usepackage[dvipsnames]{xcolor} 
\usepackage{tikz}
\usepackage{tcolorbox}
\tcbuselibrary{skins,breakable}
\tcbset{enhanced jigsaw}
\usepackage[linesnumbered,ruled,vlined]{algorithm2e}
\usepackage{nicefrac}

\usepackage[top=1in, bottom=1in, left=1in, right=1in]{geometry}
\usepackage[pdfstartview=FitH,pdfpagemode=UseNone,colorlinks,linkcolor=NavyBlue,filecolor=blue,citecolor=OliveGreen,urlcolor=NavyBlue,pagebackref]{hyperref}

\usepackage[capitalise,nameinlink]{cleveref}
\usepackage[T1]{fontenc}
\usepackage{mathtools,dsfont}

\newtheorem{lemma}{Lemma}[section]
\newtheorem{theorem}[lemma]{Theorem}
\newtheorem{definition}[lemma]{Definition}
\newtheorem{proposition}[lemma]{Proposition}
\newtheorem{observation}[lemma]{Observation}
\newtheorem{corollary}[lemma]{Corollary}

\newtheorem{fact}[lemma]{Fact}

\title{Algorithmic Improvements to List Decoding of \\Folded Reed-Solomon Codes}
\author{
Vikrant Ashvinkumar\thanks{{\tt Department of Computer Science, Rutgers University}. {\tt va264@cs.rutgers.edu}} \and
Mursalin Habib\thanks{{\tt Department of Computer Science, Rutgers University}. {\tt mursalin.habib@rutgers.edu}. Supported by the National Science Foundation under Grants CCF-2313372 and CCF-2443697.} \and
 Shashank Srivastava\thanks{{\tt DIMACS, Rutgers University \& Institute for Advanced Study}. {\tt shashanks@ias.edu}.}
 }
\date{}

\usepackage{xspace}

\newcommand{\eps}{\ensuremath{\varepsilon}}

\DeclarePairedDelimiter\parens{\lparen}{\rparen}
\DeclarePairedDelimiter\abs{\lvert}{\rvert}

\DeclarePairedDelimiter\braces{\lbrace}{\rbrace}
\DeclarePairedDelimiter\brackets{\lbrack}{\rbrack}

\DeclarePairedDelimiterXPP\lnorm[1]{}\lVert\rVert{_2}{#1}

\newcommand{\inbraces}[1]{\braces*{#1}}           
\newcommand{\insquare}[1]{\brackets*{#1}}             

\newcommand{\set}[1]{\braces*{#1}}
\newcommand{\card}[1]{\abs*{#1}}
\newcommand{\paren}[1]{\parens*{#1}}

\newcommand{\F}{\mathbb{F}}
\newcommand{\Z}{\mathbb{Z}}

\newcommand{\poly}{\mathrm{poly}}
\newcommand{\polylog}{\mathrm{polylog}}
\newcommand{\tildeO}[1]{\widetilde{O}\parens*{#1}}

\newcounter{tboxalg}
\renewcommand{\thetboxalg}{\arabic{tboxalg}} 

\newcommand{\algname}[1]{\underline{\textbf{Algorithm~\thetboxalg:}~#1}\\}

\newenvironment{tbox}{
  \refstepcounter{tboxalg}
  \begin{tcolorbox}[
    enlarge top by=5pt,
    enlarge bottom by=5pt,
    float=h,
    boxsep=2pt,
    left=5pt,
    right=7pt,
    top=10pt,
    arc=0pt,
    boxrule=1pt,toprule=1pt,
    colback=white
  ]
}{
  \end{tcolorbox}
}

\newcommand{\frs}{\mathcal{C}^{\textnormal{FRS}}\xspace}

\def\efrs{{\mathrm{Enc}^{\mathrm{FRS}}}}

\newcommand{\mper}{\,.}

\newcommand{\Esymb}{\mathbb{E}}
\newcommand{\Psymb}{\mathbb{P}}

\DeclareMathOperator*{\ExpOp}{\Esymb}

\DeclareMathOperator*{\prob}{\mathbb{P}}

\makeatletter
\def\Pr#1{%
    \ProbabilityRender{\Psymb}{#1}%
}

\def\Ex#1{%
    \ProbabilityRender{\Esymb}{#1}%
}

\def\ProbabilityRender#1#2{
  \@ifnextchar\bgroup%
  {\renderwithdist{#1}{#2}}
   {\singlervrender{#1}{#2}}
}
\def\singlervrender#1#2{%
   \ensuremath{\mathchoice
       {{#1}\insquare{ #2 }}
       {{#1}\insquare{ #2 }}
       {{#1}\insquare{ #2 }}
       {{#1}\insquare{ #2 }}
   }
}
\def\renderwithdist#1#2#3{%
   \@ifnextchar\bgroup
   {\superfancyrender{#1}{#2}{#3}}
   {\ensuremath{\mathchoice
      {\underset{#2}{#1}\insquare{ #3 }}
      {{#1}_{#2}\insquare{ #3 }}
      {{#1}_{#2}\insquare{ #3 }}
      {{#1}_{#2}\insquare{ #3 }}
     }
   }
}
\def\superfancyrender#1#2#3#4#5{
   \ensuremath{\mathchoice
      {\underset{#1}{{#1}}\left#4 #3 \right#5}
      {{#1}_{#2}#4 #3 #5}
      {{#1}_{#2}#4 #3 #5}
      {{#1}_{#2}#4 #3 #5}
   }
}
\makeatother

\newcommand{\sub}{\ensuremath{\subseteq}}

\newcommand{\defeq}{:=}  

\newcommand{\spike}[2]
{\bgroup
  \sbox0{#2}%
  \rlap{\usebox0}%
  \hspace{0.5\wd0}%
  \makebox[0pt][c]{\textcolor{lightgray}{\rule[\dimexpr \ht0+1pt]{0.5pt}{#1}}}
  \makebox[0pt][c]{\textcolor{lightgray}{\rule[\dimexpr -\dp0-#1-1pt]{0.5pt}{#1}}}
  \hspace{0.5\wd0}%
\egroup}

\newcommand{\calC}{{\mathcal C}}

\newcommand{\calH}{{\mathcal H}}

\newcommand{\calL}{{\mathcal L}}

\begin{document}
\maketitle
\begin{abstract}
	Folded Reed-Solomon (FRS) codes are a well-studied family of codes, known for achieving list decoding capacity. In this work, we give improved deterministic and randomized algorithms for list decoding FRS codes of rate $R$ up to radius $1-R-\varepsilon$.\vspace{0.1cm}

    We present a deterministic decoder that runs in near-linear time $\widetilde{O}_{\varepsilon}(n)$, improving upon the best-known runtime $n^{\Omega(1/\varepsilon)}$ for decoding FRS codes. Prior to our work, no capacity achieving code was known whose deterministic decoding could be done in time $\widetilde{O}_{\varepsilon}(n)$.\vspace{0.1cm}

    We also present a randomized decoder that runs in fully polynomial time $\mathrm{poly}(1/\varepsilon) \cdot \widetilde{O}(n)$, improving the best-known runtime $\mathrm{exp}(1/\varepsilon)\cdot \widetilde{O}(n)$ for decoding FRS codes. Again, prior to our work, no capacity achieving code was known whose decoding time depended polynomially on $1/\varepsilon$.\vspace{0.1cm}
%

    Our results are based on improved pruning procedures for finding the list of codewords inside a constant-dimensional affine subspace.
\end{abstract}

\section{Introduction}

An error-correcting code, or simply code, $\calC$ is a subset of $\Sigma^n$ for some finite alphabet $\Sigma$, and the elements of this set are called codewords. 
The code $\calC$ is supposed to be designed in a way that any two codewords in $\calC$ are far from each other in Hamming distance. The size of $\calC$ is captured by its rate $R(\calC)$, and the pairwise separation by its distance $\Delta(\calC)$, defined as:
\[
	R(\calC) \defeq \frac{\log_{|\Sigma|}{|\calC|}}{n} \qquad \text{and} \qquad \Delta(\calC) \defeq \min_{\substack{x,y\in \calC \\ x\neq y}} \Delta(x,y)
\]
where $\Delta(x,y) \in [0,1]$ denotes the normalized Hamming distance between strings $x$ and $y$. The integer $n$ is called the block length of the code.

\paragraph{Unique Decoding.} The main goal in designing codes is to tolerate errors. The error-correcting capability of a code is facilitated by its distance, since, for a code with distance \(\Delta\), as long as a codeword is corrupted in fewer than $\frac{\Delta \cdot n}{2}$ positions, one can always uniquely map the corrupted codeword back to the original codeword. This process is called \emph{unique decoding}. 
Thus, to maximize our error tolerance, we would like to design codes with as large distance as possible. 
However, as can be intuitively seen, a large distance for a code is in tension with its size. This tension is captured by the Singleton bound, which says that for any code \(\mathcal{C}\), we must have $R(\calC) + \Delta(\calC) \leq 1 +\frac{1}{n}$.

That is, for an infinite family of codes with rates lower bounded by $R$ and distances lower bounded by $\Delta$, we must have $R+\Delta\leq 1$, or $\Delta \leq 1-R$. 
Together with the above, this implies that for a rate $R$ code, the largest fraction of errors we can hope to correct is $\frac{1-R}{2}$.

\paragraph{List Decoding.} The above suggests that even for codes of very small rate, the best error correction radius cannot exceed $\frac{1}{2}$. List decoding is a relaxation where the decoder is allowed to output a \emph{list} of potential codewords that a corrupted codeword could have come from. 
Since the corruption is adversarial and we only assume an upper bound on the number of positions corrupted, this list of potential codewords is simply the list of codewords contained in a Hamming ball centered around the corrupted codeword. 
That is, given any string $g \in \Sigma^n$ and a decoding radius $\eta$, the list decoding task is to output the list of all codewords whose relative distance from $g$ is less than $\eta$. We denote this list by $\calL(g,\eta)$.

However, for this task to be possible efficiently, the list $\calL(g,\eta)$ must be of size polynomial in $n$, since otherwise it is hopeless to find a polynomial time algorithm outputting the list. This leads to the following two distinct notions of list decodability up to radius $\eta$:
\begin{itemize}
\item[-] $\calC$ is \emph{combinatorially} list decodable if for every $g$, the list $\calL(g,\eta)$ is of size polynomial in $n$.
\item[-] $\calC$ is \emph{algorithmically} list decodable if there is a polynomial time algorithm that receives $g$ as input and outputs the list $\calL(g,\eta)$.
\end{itemize}
Of course, combinatorial list decodability is a prerequisite for algorithmic list decodability.

\paragraph{List Decoding Capacity.} 
The potential of list decoding is perhaps best illustrated by random codes or random linear codes, both of which yield families of rate $R$ codes that are combinatorially list decodable up to radii arbitrarily close to $1-R$. On the other hand, it is known that any code of rate $R$ cannot be list decoded up to a radius larger than $1-R$. 
Therefore, this threshold $1-R$ is called the list decoding capacity.\footnote{We remark that one can establish better upper bounds than $1-R$ on the list decoding of rate $R$ codes over small alphabets such as binary codes, but in this work we are only concerned with the large alphabet regime.}

While many random ensembles of codes offer combinatorial list decodability up to the capacity $1-R$, we do not know of any efficient algorithms to decode them. 
Explicit constructions of codes matching the combinatorial behavior of random codes often pave the path to efficient decodability, as algorithms can exploit the structure exposed by the combinatorial proof. 
Furthermore, constructing such codes explicitly also becomes an important pseudorandomness challenge, and they can often be used as pseudorandom primitives in other applications \cite{GUV09, Vadhan12, FG15}.

\paragraph{Explicit Constructions.} The well-studied Reed-Solomon family of codes yields rate $R$ codes with distance $1-R$. Each codeword of a Reed-Solomon code is obtained by evaluating a polynomial in $\F_q[X]$ of degree less than $Rn$ on $n$ distinct points in $\F_q$. 
This means that $q$ must be at least $n$, and typically one chooses $q= \Theta(n)$. Further, these codes can be uniquely decoded up to radius $\frac{1-R}{2}$ in near-linear time~\cite{reed1978fast}.

The works by Sudan~\cite{Sudan97} and by Guruswami and Sudan~\cite{GS99} give algorithms to decode Reed-Solomon codes up to radius $1-\sqrt{R}$. By choosing $R$ to be small, this allows error correction from a fraction of errors that approaches 1, a fact that led to several applications in complexity theory~\cite{Gur06}. 
This remains the best radius for algorithmic list decoding of Reed-Solomon codes, and for a while, it remained open to design explicit codes with rate $R$ that could be decoded close to the capacity $1-R$.

\begin{sloppypar}
The first construction of codes achieving list decoding capacity was achieved by Guruswami and Rudra~\cite{GR08}, building on the work of Parvaresh and Vardy~\cite{PV05}, using Folded Reed-Solomon (FRS) codes. 
They used a careful modification of the Reed-Solomon code to give for any $\eps > 0$, a code of rate $R$ and distance $1-R$, that is list decodable up to a radius $1-R-\eps$. 
The alphabet size for the code is $n^{O(1/\eps^2)}$, the list produced is of size at most $n^{O(\nicefrac{1}{\eps})}$, and the decoding can be done in time $n^{O(\nicefrac{1}{\eps})}$.
\end{sloppypar}

\paragraph{Algorithmic challenges.} A lot of work has since been done on trying to improve the alphabet size, list size, and decoding time above. Here, we focus on progress made towards improving the decoding time for FRS codes. 
Note that the decoder must output the entire list, which means that its runtime is lower bounded by the list size. Therefore, any improvement to the runtime must also prove better list size bounds.

The first progress in this direction was made by Guruswami~\cite{Gur11} who, building on ideas of Vadhan~\cite{Vadhan12}, showed that the list of codewords is always contained in an affine subspace of dimension $O(\nicefrac{1}{\eps})$, and that a basis for it can be found in time $(\nicefrac{1}{\eps})^{O(1)}\cdot O(n^2)$. 
Note that the list of codewords could still be of size $q^{\nicefrac{1}{\eps}} \approx n^{\nicefrac{1}{\eps}}$ in the worst case, and even algorithmically searching over the subspace takes time $\Omega(n^{\nicefrac{1}{\eps}})$.

\begin{sloppypar}
Nevertheless, this structural characterization was used by Kopparty, Ron-Zewi, Saraf, and Wootters~\cite{KRSW23} to show that the list size above can be made $(\nicefrac{1}{\eps})^{O(\nicefrac{1}{\eps})}$. In fact, their proof was algorithmic: they gave a randomized algorithm such that any codeword in the list is output with probability at least $\eps^{O(\nicefrac{1}{\eps})}$, which immediately implies the list size bound. This proof was subsequently simplified and tightened by Tamo~\cite{Tamo24}.
\end{sloppypar}

Finally, two concurrent works from last year, by the third author~\cite{Sri25} and by Chen and Zhang~\cite{CZ25} brought down the list size to $O(\nicefrac{1}{\eps^2})$ and $O(\nicefrac{1}{\eps})$ respectively. Note that when decoding up to radius $1-R-\eps$, the list must be of size at least $\Omega(\nicefrac{1}{\eps})$ by the generalized Singleton bound \cite{ST20}, and the Chen-Zhang result matches it optimally. The interested reader is referred to the exposition by Garg, Harsha, Kumar, Saptharishi, and Shankar~\cite{GHKSS25} that covers these recent improvements.

\subsection{Our Results}
\label{sec:our-results}

The algorithm for decoding FRS codes is typically implemented as a two-step process. First, given the corrupted codeword, the algorithm finds a basis for an affine subspace that contains the list. Second, the algorithm searches over this affine subspace to discover the final list. As mentioned before, it was shown in \cite{Gur11} how to implement the first step in time $(\nicefrac{1}{\eps})^{O(1)}\cdot O(n^2)$, which was recently improved by Goyal, Harsha, Kumar and Shankar~\cite{GHKS24} to time $(\nicefrac{1}{\eps})^{O(1)}\cdot \tildeO{n}$.

In this work, we focus on the second step -- pruning the subspace to find the list. The only non-trivial algorithm for this task is the one introduced in~\cite{KRSW23},\footnote{In fact, the KRSW algorithm works for any distance $\Delta$ linear code when decoding up to radius $\Delta-\eps$.} which we summarize next. Given a basis for an affine subspace $\calH$ of dimension $k>0$ that contains the list $\calL(g,1-R-\eps)$, the KRSW algorithm picks a coordinate $i\in [n]$ uniformly at random, and considers the new affine subspace\footnote{$\calH_i$ could also be empty, but that doesn't affect the argument.} 
\[
	\calH_i = \{f\in \calH ~:~ f(i) = g(i)\} \mper
\]
Using the distance of the code, it is easy to show that $\Pr{i\in [n]}{\dim(\calH_i) < k} \geq 1-R$, and for any $h\in \calL(g,1-R-\eps)$, 
\[
	\Pr{i\in [n]}{h\in \calH_i} \geq R+\eps
\]
The algorithm then recurses on $\calH_i$, which, with probability at least $\eps$, both contains $h$ and is of dimension at most $k-1$. This continues until the dimension reduces to zero. With probability at least $\eps^k$, this 0-dimensional affine subspace must contain $h$.

For FRS codes, we can choose $k=O(\nicefrac{1}{\eps})$, and that gives a success probability of at least $\eps^{O(\nicefrac{1}{\eps})}$. This algorithm is then repeated $2^{\tildeO{\nicefrac{1}{\eps}}}$ times so that every list codeword is discovered with high probability.
The KRSW algorithm has the following two shortcomings:
\begin{itemize}
	\item The KRSW algorithm is randomized. The only deterministic algorithm we know is the trivial one, which takes $n^{O(\nicefrac{1}{\eps})}$ time, despite the list being of size only $O(\nicefrac{1}{\eps})$. Can we get a deterministic algorithm that runs in time $F(\eps)\cdot n^{O(1)}$, with the exponent of $n$ being a fixed constant independent of $\eps$, and $F(\eps)$ being some function of $\eps$? This question was mentioned as an open problem in \cite{Sri25, CZ25, LLMWX25}.
	\item The runtime of the KRSW algorithm depends exponentially on $\nicefrac{1}{\eps}$. 
	With the list size improvements from~\cite{Sri25, CZ25}, one could expect a pruning algorithm that runs in time $O( n + \nicefrac{1}{\eps})$, or at least a \emph{fully polynomial time} algorithm that runs in time $(n+\nicefrac{1}{\eps})^{O(1)}$. 
	This would allow the use of FRS as capacity achieving codes even in the regime when $\eps$ is subconstant. This question appears as Open Problem 17.4.3 in the book by Guruswami, Rudra, and Sudan~\cite{GRS23}, and was also asked in \cite{GR21, GHKSS25, BCDZ25}.
\end{itemize}
In this work, we resolve both of these problems.
Note that the first question makes sense while pruning for general linear codes, while the second question only makes sense for FRS codes. This is because a $\poly(\nicefrac{1}{\eps})$ list size is only known for the special case of FRS codes.
\subsubsection{Deterministic Pruning}

Our main result in this direction is a deterministic procedure that can prune a constant-dimensional subspace in near-linear time. More precisely, for any linear code of block length $n$, alphabet size $q$, and distance $\Delta$, the procedure takes as input a basis for a $k$-dimensional affine subspace containing the list $\calL(g,\Delta-\eps)$, and outputs the list $\calL(g,\Delta-\eps)$ in time $(\nicefrac{1}{\eps})^{2^k}\cdot O(n)\cdot \polylog(q)$.

Combining this with the results of \cite{Gur11} and \cite{GHKS24}, we get the following result for deterministic decoding of FRS codes.
\begin{theorem}[Informal version of \cref{cor:deterministic-frs}]
	For all \(\varepsilon> 0\) and \(R \in (0, 1)\), there exists an infinite family of FRS codes of rate \(R\) that can be list decoded up to radius \(1 - R - \varepsilon\) in time \(\parens*{\nicefrac{1}{\eps}}^{O(2^{1/\eps})} \cdot \widetilde{O}(n)\) by a deterministic algorithm, where $n$ is the block length of the code.
%
\end{theorem}

This gives the first family of codes that is deterministically decodable to capacity in near-linear time $\tildeO{n}$. The previous best deterministic algorithm from \cite{KRRSS21} could only achieve a runtime of $n^{1+o(1)}$, and was based on a much more complicated construction than usual FRS codes.

Beyond parameters, we remark that the algorithmic framework for list decoding pioneered by Guruswami~\cite{Gur11} and Guruswami-Wang~\cite{GW13} works by first establishing a set of linear constraints such that any list codeword must satisfy them, and the list is then extracted out of the solution to this linear system. This framework has no way of distinguishing between two true list codewords and their affine combinations, even if the affine combinations no longer lie in the Hamming ball. Thus, under this popular decoding framework \cite{BHKS23}, the task of searching over the subspace for the true list (which is of much smaller size than the search space) seems unavoidable.

On the other hand, for any subspace of dimension $k$, the only deterministic algorithm we know to search over this space is based on a brute force search, which takes time $\Omega(n\cdot q^k)$. Even for $k=1$, which is necessary for any decoding radius beyond $\frac{1-R}{2}$, this translates to a quadratic time algorithm. 

\paragraph{Attempt 1.} In \cite{Sri25}, an approach was suggested to deal with the $k=1$ case in near-linear time. Given a 1-dimensional subspace $\calH$, the idea is to compute for each coordinate $i \in [n]$, the set $\calH_i$ of codewords in $\calH$ that agree with $g$ on the $i^{\textnormal{th}}$ coordinate. That is,
\[
	\calH_i = \{f\in \calH ~:~ f(i) = g(i)\} \mper
\]
As argued before for the KRSW algorithm, it must be the case that for any $h\in \calL(g,1-R-\eps)$, $\calH_i = \{h\}$ in at least $\eps$ fraction of positions. At this point, one can simply look for the $\nicefrac{1}{\eps}$ most frequent codewords that appear among the 0-dimensional spaces in $\{\calH_i\}_{i\in [n]}$.

However, this strategy fails when the dimension $k$ becomes 2. The issue is that while the KRSW argument will still say that there is an $\eps$ fraction of coordinates where $\dim(\calH_i) < 2$ and $h\in \calH_i$, these $\calH_i$ containing $h$ need not be the same subspace. If there were a single 1-dimensional subspace (line) containing $h$ in most of these coordinates, then we could have restricted our search to such \emph{popular} lines. However, there are $q$ many 1-dimensional affine subspaces (lines) containing $h$, and the $\eps n$ coordinates could have corresponded to different lines, all passing through $h$.

\paragraph{Attempt 2.} The other approach one could consider is to use expander edges to derandomize the KRSW argument. Recall that to prune a $2$-dimensional space, the KRSW algorithm chooses $2$ independent coordinates from $[n]$, or equivalently, a uniformly random sample from $[n]^2$, and asks whether there is a unique codeword agreeing with $g$ on both these coordinates. It is a standard derandomization technique to replace $[n]^2$ with the set of edges of a spectral expander graph on $n$ vertices. These graphs can be made to have only $O(n)$ edges, while still having the property that $\Pr{(u,v)\in E}{u\in S, v\in T} \approx \frac{|S|}{n}\cdot\frac{|T|}{n}$ for any pair of subsets \(S, T\) of vertices.

Fix an $h \in \calL(g,1-R-\eps)$, and define $S_h = \{i\in [n]:\dim(\calH_i)<2, h(i) = g(i)\}$. As before $|S_h|\geq \eps \cdot n$. Suppose it holds that the lines $\{\calH_i\}_{i\in S_h}$ are all distinct from each other. Then for any edge $(u,v)$ such that $u,v\in S_h$ but $u\neq v$, it holds that $\calH_u \cap \calH_v = \{h\}$. By the expander property, the fraction of such edges is $\Omega(\eps^2)$, and we can quickly find $\{h\}$ as a popular labeling if we label each edge $(u,v)$ with $\calH_u\cap\calH_v$. Such a labeling can be done in time $O(n)$ because the number of edges is $O(n)$.

\paragraph{Final Algorithm.} The two approaches above only seem to work assuming few distinct lines in $S_h$ or many distinct lines in $S_h$, respectively. Our final algorithm is a win-win analysis, that combines both these approaches by first recursively decoding on each popular line, and then it uses expander edges on the remaining coordinates to see if $h$ pops out as a 0-dimensional space on the edges.

We remark that naively replacing all pairs $[n]^2$ by expander edges in the KRSW argument does not seem easy to analyze. Suppose we wish to look at the number of edges $(u,v)\in S_h \times S_h$ such that $\calH_u \cap \calH_v = \{h\}$. For any $u\in S_h$, let $T_{hu}$ be the set of $v\in S_h$ with this property that $\calH_u \cap \calH_v = \{h\}$.  The key issue is that this set $T_{hu}$ changes with $u$.

From the KRSW argument, we know that for each $u$, the set $T_{hu}$ has size at least $\eps n$, which is sufficient for a randomized algorithm to succeed. But for a deterministic algorithm, where each $u$ is only connected to constantly many vertices in the sparse expander, it could be that the neighborhood of every $u$ is disjoint from corresponding $T_{hu}$, and so we never see $\{h\}$ written on any expander edge.

\subsubsection{Fully Polynomial Time Pruning}

Our second main result gives a randomized algorithm for decoding FRS codes that runs in time $\poly(\nicefrac{1}{\eps})\cdot\tildeO{n}$.
\footnote{By the generalized Singleton bound \cite{ST20, Rot24}, the list size at decoding radius $1-R-\eps$ must be of size $\Omega(\nicefrac{1}{\eps})$, which means $\Omega(n+\nicefrac{1}{\eps})$ is a natural lower bound on any list decoder's running time.}

\begin{theorem}[Informal version of \cref{cor:randomized_frs}]
	For all \(\varepsilon> 0\) and \(R \in (0, 1)\), there exists an infinite family of FRS codes of rate \(R\) that can be list decoded up to radius \(1 - R - \varepsilon\) in time \(\poly\parens*{\frac{1}{\eps}} \cdot \widetilde{O}(n)\) by a randomized algorithm, where $n$ is the block length of the code.
\end{theorem}

To the best of our knowledge, this is the first example of a code achieving list decoding capacity that can be decoded in time $\poly(n,\nicefrac{1}{\eps})$.  This allows for the use of capacity-achieving codes even in the regime when $\eps$ is needed to be subconstant, allowing for list decoding from a $\left(1-\frac{1}{\poly(n)}\right)$ radius in $\poly(n)$ time.

The bottleneck in previous decoders \cite{KRSW23, GR21} has always been the pruning step, and we show that for FRS codes, this pruning can be done in time that depends polynomially on $\nicefrac{1}{\eps}$. Of course, such an algorithm also implicitly proves that the list size depends polynomially on $\nicefrac{1}{\eps}$. Since the only two known proofs of $\poly(\nicefrac{1}{\eps})$ list size for FRS codes are \cite{Sri25} and \cite{CZ25}, it is not surprising that our algorithm is heavily based on \cite{Sri25}. 

Suppose we are given an integer $s\geq 1$. Then, provided that the folding parameter of the code is sufficiently large, our algorithm can decode up to radii arbitrarily close to $\frac{s}{s+1}(1-R)$ by outputting \emph{any} codeword in the list with probability at least $1/s^2$. This implies that the list must be of size at most $s^2$, and repeating this process allows us to discover every element in the list with high probability. Finally, choosing $s\approx \nicefrac{1}{\eps}$ allows us to decode near the list decoding capacity. 

We note that for small values of $s$ such as $s=2$, when by the Chen-Zhang result we know that the list contains at most 2 codewords, our algorithm outputs each of these codewords with probability at least $\nicefrac{1}{4}$. In particular, this probability is independent of the gap to the radius $\frac{2}{3}(1-R)$. The reader is referred to \cref{thm:randomized-success} for a technical statement. 

We believe that this algorithm could be practical and is easy to implement. Further, the gains in the decoding radius are ``front-loaded'', with the biggest increase happening for small values of $s$. While we need a large folding parameter to approach the radius $\frac{s}{s+1}(1-R)$, even a small folding parameter (say, $s$-folded) will allow for a radius beyond the unique decoding threshold $\frac{1-R}{2}$ for reasonably low rates.

We further note that a recent work of Brakensiek, Chen, Dhar, and Zhang~\cite{BCDZ25} also gives an algorithm for decoding FRS codes with a polynomial dependence on $\nicefrac{1}{\eps}$. However, that result is only applicable for semi-adversarial errors. For such semi-adversarial errors, the list is likely to contain a unique codeword even for large decoding radii, and the question of pruning does not arise. In contrast, our result shows that the exponential dependence on $\nicefrac{1}{\eps}$ in previous decoders was not necessary, even in the fully adversarial error model.

\paragraph{Proof Technique.}

Our algorithm is based on turning the combinatorial proof of $O(\nicefrac{1}{\eps^2})$ list size in \cite{Sri25} into an efficient algorithm. The proof in \cite{Sri25} is itself an improved combinatorial analysis of the KRSW argument in the setting of FRS codes. The key idea is that conditioning a coordinate to be an agreement is likely to reduce the dimension of the search space by much more than 1, due to the folded nature of the code. In fact, roughly speaking, it is known from the work of  Guruswami and Kopparty~\cite{GK16} that $\Ex{i\in [n]}{\dim(\calH_i)} \leq R\cdot \dim(\mathcal{H})$.

Inspired by this combinatorial proof, we make the following three changes to the KRSW algorithm:
\begin{enumerate}[(i)]
	\item The key idea in \cite{Sri25} was a stronger inductive hypothesis with the list size depending linearly on $k$ instead of exponentially. This allowed the inductive step to directly use the \cite{GK16} result on expected dimension, instead of a tail bound derived from it. Inspired by that, we strengthen the inductive hypothesis so that our algorithmic procedure, when searching over a $k$-dimensional affine subspace, outputs any $h$ in the list with probability at least $\eps/k$. Note that this is significantly better than the $\eps^k$ bound in the KRSW algorithm.
	\item To make the inductive step work however, we can no longer pick coordinates uniformly at random as in the KRSW algorithm. Instead, we pick coordinates based on the distribution of $\dim(\calH_i)$ over $i$.
	\item Finally, recall that the KRSW algorithm picks multiple coordinates and hopes that the dimension reaches zero. Our algorithm is recursive in nature, where it picks only one coordinate $i$ and issues a recursive call to prune the resulting subspace $\calH_i$, until the dimension of the search space hits zero. This means that if one were to unroll the recursion, the different coordinates are picked according to different distributions. In fact, if the first coordinate is $i$, the distribution of the second coordinate depends on $i$ (actually, only on $\calH_i$).
\end{enumerate}

We think of the non-uniform and adaptive procedure described in (ii) and (iii) above as the key insight in our proof. With these changes, our proof largely follows the one in \cite{Sri25}.

\subsection{Related Work}

Since the work of Guruswami and Rudra~\cite{GR08} showing that FRS codes can be list decoded up to capacity, there has been a lot of work on obtaining improvements in various parameters. A non-exhaustive list includes works that use algebraic-geometric (AG) codes~\cite{Gur09, GX12, GX22, GR21}, subspace-evasive sets \cite{DL12, GX13, GK16, GR21}, and tensoring \cite{GGR09, HRW20, KRRSS21}. In particular, Guo and Ron-Zewi~\cite{GR21} were able to achieve a fully polynomial time decoding algorithm, except for the final pruning step, which was still exponential in $\nicefrac{1}{\eps}$.

Several other code families, beyond FRS and AG codes, have also been shown to achieve list decoding capacity, with most of these relying on the interpolation machinery from \cite{Gur11}. This includes univariate multiplicity codes \cite{GW13, Kop15}, (bivariate) linear operator codes \cite{BHKS23, PZ25}, and subcodes of Reed-Solomon codes \cite{GX13, BST25}. Given the common proof ingredients, it is possible that our pruning algorithms may be applicable to some of these codes.

Among these, univariate multiplicity codes share many properties with FRS codes, and all the results we use \cite{Gur11, GW13, KRSW23, GHKS24, Sri25, CZ25} also extend to univariate multiplicity codes. Our algorithms also extend to give improved deterministic and randomized decoders for these univariate multiplicity codes, but we choose to focus on the FRS codes for the sake of brevity.

As mentioned before, it was only recently shown \cite{Sri25, CZ25} that FRS codes have a list of size $\poly(\nicefrac{1}{\eps})$, unlocking the fully polynomial time algorithm of this work. Since then, Jeronimo, Mittal, Srivastava, and Tulsiani~\cite{JMST25} were able to obtain purely combinatorial constructions of capacity-achieving codes based on expander graphs that also have the optimal list size $O(\nicefrac{1}{\eps})$. This matches the FRS codes, and in fact, its alphabet size is a constant independent of the block length. A near-linear time randomized\footnote{The algorithms in \cite{ST25} can also be made deterministic, but then they no longer remain near-linear time.} algorithm to decode these expander-based codes was also found in~\cite{ST25}, but this algorithm does not depend polynomially on $\nicefrac{1}{\eps}$.

\begin{sloppypar}
  Finally, we mention that optimal $O(\nicefrac{1}{\eps})$ list size is also known for several random ensembles of codes, such as randomly evaluated Reed-Solomon codes~\cite{ST20, BGM23, GZ23, AGL24}, random linear codes~\cite{AGL24}, random AG codes~\cite{BDG24, BDGZ24}, etc. However, no efficient algorithms are known for these random codes.  
\end{sloppypar}

\subsection{Discussion}

An obvious open problem left by our work is to get a fully polynomial time \emph{deterministic} algorithm for decoding FRS codes. As pointed out above, our randomized algorithm can be seen as picking coordinates non-uniformly and adaptively, and it would be interesting to see a derandomization of this strategy. Another interesting question is whether there is a natural derandomization of the KRSW algorithm based on the algebraic structure of affine subspaces and linear codes, instead of based on expander graphs. That is, can we use this algebraic structure to guide the choice of coordinates to condition on?

The proof in \cite{Sri25} that we algorithmize starts from the assertion of \cite{Gur11, GW13} that the list is contained in a constant-dimensional affine subspace, which itself is based on certain objects called folded Wronskian determinants. The pruning steps in \cite{Sri25} and our work also depend on folded Wronskian determinants via the Guruswami-Kopparty lemma \cite{GK16} on the average dimension (see \cref{lemma: GK16}). In contrast, the Chen-Zhang result \cite{CZ25} works directly with folded Wronskian to get the sharp $O(\nicefrac{1}{\eps})$ list size in one shot. Is there a way to directly algorithmize the Chen-Zhang result?

Continuing the above, as long as our algorithms are only based on solving linear systems of equations, it is unavoidable to deal with linear/affine subspaces.
One way to avoid the pruning procedure is to flip the order of the two steps, where we first randomly pick constantly many coordinates, and then look for the hopefully \emph{unique} codeword that agrees with the received word in all those coordinates. For starters, can we make the KRSW argument work in this manner? It could be of use that the KRSW choice of coordinates is agnostic to the affine subspace we are trying to prune.

One way to convert a list decoding problem into a unique decoding problem is to convert some errors into \emph{erasures}, as was used heavily by Jeronimo \textit{et al}.~\cite{JMST25}. That is, if we could fix the locations of the $1-R-\eps$ errors, then there is at most one codeword consistent with such an error pattern. Of course, there are many ways to fix these locations, and this corresponds to the many codewords in list decoding. Can we view the KRSW style conditioning on random coordinates as allowing us to convert errors into erasures, thereby providing a more transparent and combinatorial view into the list decoding behavior of algebraic codes?
\section{Preliminaries}

\paragraph{Error-Correcting Codes.} We begin this section by reviewing some basic definitions about error-correcting codes.
\begin{definition}[Distance]
	Let $\Sigma$ be a finite alphabet and let $f,g\in \Sigma^n$. Then the (normalized) Hamming distance
        between $f,g$ is defined as \[\Delta(f,g) = \prob_{i\sim [n]}[{f(i) \neq g(i)}]. \]
\end{definition}

Throughout the entire paper, we will always be working with the \emph{normalized} distance as above.
\begin{definition}[Code, distance and rate]
	A code $\calC$ of block length $n$, distance $\Delta$ and rate $R$ over an alphabet $\Sigma$ is a set $\calC \subseteq \Sigma^n$ with the following properties
	\begin{enumerate}[(i)]
		\item $R = \frac{\log_{|\Sigma|} |\calC|}{n}$
		\item $\Delta = \min_{\substack{h_1,h_2\in \calC \\ h_1\neq h_2}} \Delta(h_1,h_2)$
	\end{enumerate}
\end{definition}

\begin{definition}[$\F_q$-Linear Codes]
    A code $\calC$ of block length $N$ is said to be $\F_q$-linear if $\Sigma = \F_q^m$ for some prime power $q$, and $\calC$ is a linear subspace of the linear space $\Sigma^N$ over $\F_q$.
\end{definition}

\begin{definition}[List of codewords]
	Let $\calC$ be a code with alphabet $\Sigma$ and block length $n$. Given $g\in \Sigma^n$, we use $\calL(g,\eta)$ to denote the list of codewords from $\calC$ whose distance from $g$ is less than $\eta$. That is,
	\[
		\calL(g,\eta) = \inbraces{ h\in \calC : \Delta(g,h) <\eta}\mper
	\]
\end{definition}

	We say that a code is combinatorially list decodable up to radius $\eta$ if for every $g\in \Sigma^n$, $\calL(g,\eta)$ is of size at most $\poly(n)$. Likewise, we say a code is algorithmically list decodable up to radius $\eta$ if it is combinatorially list decodable up to $\eta$, and the list $\calL(g,\eta)$ can be found in time $\poly(n)$.

    \paragraph{Folded Reed-Solomon Codes.} Our main results concern Folded-Reed Solomon codes, which we now formally define.  
\begin{definition}[Folded Reed-Solomon Codes]
	Let $\F_q$ be a field with $q>n$, and $\gamma$ be an element of order at least $n$. The encoding function for the $m$-folded Reed-Solomon code $\frs$ with rate $R$, blocklength $N = n/m$, alphabet $\F_q^m$, and distance $1-R$ is denoted as $\efrs : \F_q[X]^{<Rn} \rightarrow (\F_q^m)^N$, given by
	\[
		f(X) \mapsto \insquare{ \begin{pmatrix}
	f(1) \\
	f(\gamma)\\
	\vdots \\
	f(\gamma^{m-1})
	\end{pmatrix}, 
	\begin{pmatrix}
	f(\gamma^m) \\
	f(\gamma^{m+1})\\
	\vdots \\
	f(\gamma^{2m-1})
	\end{pmatrix},
	\cdots ,\begin{pmatrix}
	f(\gamma^{n-m}) \\
	f(\gamma^{n-m+1})\\
	\vdots \\
	f(\gamma^{n-1})
	\end{pmatrix} } \in (\F_q^m)^{N}
	\]
	The code $\frs$ is given by
	\[
		\frs = \inbraces{ \efrs(f(X)) : f(X) \in \F_q[X]^{<Rn}}.
	\]
\end{definition}

Throughout this paper, we will rely on the following structural theorem due to Guruswami~\cite{Gur11}, which asserts that the list of codewords close to any received word for the folded Reed-Solomon code lies within a low-dimensional affine subspace.

\begin{theorem}[\cite{Gur11}]
\label{thm:list-in-affine-subspace}
Let \(\mathcal{C}^{\textnormal{FRS}}\) be an \(m\)-folded Reed-Solomon code of blocklength \(N = n/m\) and rate \(R\). For any integer \(s\) satisfying \(1 \leq s \leq m\), and any received word \(g \in (\mathbb{F}_q^m)^N\), there exists an affine subspace \(\mathcal{H} \subseteq \mathcal{C}^{\textnormal{FRS}}\) over \(\mathbb{F}_q\) of dimension at most \(s - 1\) such that
\[
\mathcal{L}\left(g, \frac{s}{s+1} \left(1 - \frac{mR}{m - s + 1}\right)\right) \subseteq \mathcal{H}.
\]
\end{theorem}

In our applications, the existence of such a subspace \(\mathcal{H}\) is not sufficient -- we also require an efficient method to compute it. The following result, due to Goyal \textit{et al.}~\cite{GHKS24}, guarantees the existence of a near-linear time algorithm to find a basis for this subspace.

\begin{theorem}[\cite{GHKS24}]
\label{thm:fast-affine-subspace-basis}
Let \(g \in (\mathbb{F}_q^m)^N\) be a received word, and let \(s\) be an integer with \(1 \leq s \leq m\). Let \(\mathcal{H} \subseteq \mathcal{C}^{\textnormal{FRS}}\) be the affine subspace guaranteed by~\cref{thm:list-in-affine-subspace}, satisfying
\[
\mathcal{L}\left(g, \frac{s}{s+1} \left(1 - \frac{mR}{m - s + 1}\right)\right) \subseteq \mathcal{H}.
\]
Then a basis \((h_0, h_1, \ldots, h_{s-1})\) for \(\mathcal{H}\) can be computed using at most \(\widetilde{O}(mN)\cdot \textnormal{poly}(s)\) field operations over \(\mathbb{F}_q\), where \(\mathcal{H} = \left\{h_0 + \sum_{i=1}^{s-1} \alpha_i h_i : \alpha_i \in \mathbb{F}_q \text{ for all } i \right\}\).
\end{theorem}
{Strictly speaking, the algorithm in~\cite{GHKS24} outputs a basis of \(\efrs^{-1}(\mathcal{H})\) by giving the coefficients of \(k\) polynomials in \(\mathbb{F}_q[X]^{<Rn}\). However, we can simply encode these polynomials using \(\efrs\) with an additional \(s\cdot \widetilde{O}(mN)\) field operations~\cite{fiduccia1972polynomial}.} 

Note that \cref{thm:list-in-affine-subspace} already implies a list size bound of at most \(q^{s-1}\). The following result due to Chen and Zhang~\cite{CZ25} asserts that the list is much smaller.

\begin{theorem}
    \label{thm:FRS-gen-singleton=bound}
    Let \(\mathcal{C}^{\textnormal{FRS}}\) be an \(m\)-folded Reed-Solomon code of blocklength \(N = n/m\) and rate \(R\). For any integer \(s\) satisfying \(1 \leq s \leq m\), and any received word \(g \in (\mathbb{F}_q^m)^N\), we have:
    \[\left\lvert \mathcal{L}\left(g, \frac{s}{s+1} \left(1 - \frac{mR}{m - s + 1}\right)\right)\right\rvert\leq s.\]
\end{theorem}

\paragraph{Expanders. }
Our deterministic algorithm uses expanders and the expander mixing lemma, defined below, as derandomization tools.

\begin{definition}[Expanders]
    A $d$-regular graph $G$ on $n$ vertices is an $(n,d,\lambda)$-expander if the magnitude of eigenvalues $1 = \card{\lambda_1} \ge \card{\lambda_2} \ge \ldots \ge \card{\lambda_n}$ of its normalized adjacency matrix $\frac{1}{d}A_G$ satisfy $\lambda \ge \card{\lambda_i}$ for all $i > 1$.
\end{definition}

\begin{lemma}[Expander Mixing Lemma]
\label{lem:expander-mixing}
    Let $G$ be an $(n,d,\lambda)$-expander.
    Then
    \begin{align*}
        \card{\Ex{(u,v) \sim E(G)}{f(u)g(v)} - \Ex{(u,v) \sim V(G)^2}{f(u)g(v)}} \le \lambda \sqrt{\Ex{u \sim V(G)}{f^2(u)}}\sqrt{\Ex{v \sim V(G)}{g^2(v)}},
    \end{align*}
    where $f,g : V(G) \rightarrow \mathbb{R}$ are real-valued functions. In many applications, $f,g$ are indicator functions, so that their range is $\{0,1\}$.  
\end{lemma}

We specifically use the following family of expanders.

\begin{proposition}[Gaber-Gallil Expander \cite{GG81}]
\label{prop:gg}
    The Gaber-Gallil expander is an $(n,8,\lambda)$-expander where $\lambda \le 5\sqrt{2}/8$.
    Each edge of the expander can be found using $O(1)$ operations in $\Z_{\sqrt{n}}$.
\end{proposition}

\begin{observation}[Powered Gaber-Gallil Expander]
\label{obs:gg-power}
    For any $\lambda > 0$, an $(n, \poly(\lambda^{-1}), \lambda)$-expander $G$ can be constructed in time $O(\log(\lambda^{-1})\card{E(G)})$.
\end{observation}
\begin{proof}
    Let $G_0$ be the Gaber-Gallil $(n, d_0=8, \lambda_0 \le 5\sqrt{2}/8)$-expander.
    Observe that $d_0 = O(\lambda_0^{-17}) = \textrm{poly}(\lambda_0^{-1})$.
    We can thus take $G$ to be the $\Theta(\log\lambda^{-1})$-th power of $G_0$.
    The edges of $G$ comprise $\Theta(\log\lambda^{-1})$ step walks in $G_0$, each of which can be found using $O(\log\lambda^{-1})$ operations in $\Z_{\sqrt{n}}$.
\end{proof}

\paragraph{Miscellanea. }
We implicitly use the following basic fact about affine spaces and algorithmically computing the reduced row echelon form (RREF) of matrices.

\begin{fact}
\label{fact:affine-intersection}
    Let $\calH, \calH'$ be affine spaces.
    Then $\calH \cap \calH'$ is either an affine space or empty.
\end{fact}
    

\begin{proposition}
\label{prop:rref}
    Let $M$ be an $m$ by $n$ matrix over a field.
    The Reduced Row Echelon Form of $M$ can be computed in $O(mn^2)$ time using $O(mn^2)$ field operations.
\end{proposition}

Finally, we use the following simple observation when analyzing our deterministic algorithm.

\begin{observation}
\label{obs:event-partition}
    Let $E_1, E_2, \ldots, E_t$ be disjoint events such that $\Pr{\bigcup_{i \in [t]} E_i} = p$ and $\Pr{E_i} \le p/2$ for all $i \in [t]$.
    There is a partition $X \sqcup Y = [t]$ such that $p/4 \le \Pr{\bigcup_{i \in Z} E_i} \le 3p/4$ for all $Z \in \set{X,Y}$.
\end{observation}
\begin{proof}
    We may assume without loss of generality that all events occur with positive probability. 
    
    A partition $X \sqcup Y = [t]$ such that $\max(\Pr{\bigcup_{i \in X} E_i}, \Pr{\bigcup_{i \in Y} E_i})$ is minimized satisfies the claim.
    To see this, suppose for a contradiction that $\max(\Pr{\bigcup_{i \in X} E_i}, \Pr{\bigcup_{i \in Y} E_i}) = \Pr{\bigcup_{i \in X} E_i} > 3p/4$.
    We can move any one index from $X$ to $Y$ to contradict the minimality of $\max(\Pr{\bigcup_{i \in X} E_i}, \Pr{\bigcup_{i \in Y} E_i})$.
\end{proof}

\section{A Deterministic Pruning Algorithm}\label{sec:deterministic}

In this section, we show a deterministic algorithm for list decoding FRS codes up to capacity in near-linear time.
The main technical part is a near-linear time deterministic pruning algorithm for any $\Delta$ distance $\F_q$-linear code $\calC \subseteq (\F_q^m)^N$ that -- when given a received word $g \in (\F_q^m)^N$, a fixed constant $\eps \in (0,\Delta)$, and an affine subspace $\calH \subseteq \calC$ over $\F_q$ -- computes all codewords in $\calH \cap \calL(g, \Delta-\eps)$.

\paragraph{Warm-Up 1: Pruning in the Affine Line. }
Let us begin, as a warm-up, with a procedure for pruning when $\dim(\calH) = 1$.
Assume for simplicity $m = 1$, and that all field operations and comparisons cost one unit of time.
Let $\calH = \set{h_0 + \alpha_1 h_1: \alpha_1 \in \F_q}$.
Naively, one could test if $h_0 + \alpha_1 h_1 \in \calL(g, \Delta-\eps)$ for each $\alpha_1 \in \F_q$, outputting that codeword if so.
This would take $O(Nq)$ time, which, for an RS code, is essentially $O(N^2)$ time.

We can do better with an $O(N/\eps)$ time algorithm as follows.
To each coordinate $i \in [N]$, we associate the space $\calH_i$ defined by the linear equality $(h_0 + \alpha_1 h_1)(i) = g(i)$ where $\alpha_1$ is the variable.
On the set of indices $S = \set{i: h_1(i) \neq 0}$, these spaces are $0$ dimensional and contain only the codeword $h_0 + \alpha_i h_1$ with $\alpha_i = \frac{g(i)-h_0(i)}{h_1(i)}$.
Since $\card{S} \ge \Delta N$, every codeword in $\calH \cap \calL(g, \Delta-\eps)$ must agree with $g$ on at least $\eps N$ coordinates in $S$.
If we hence let each coordinate $i \in S$ ``vote'' for the $0$ dimensional space associated to it, this says that every codeword in $\calH \cap \calL(g, \Delta-\eps)$ receives at least $\eps N$ votes.
Put otherwise, the codewords (or rather spaces) are \emph{$\eps$-popular} with respect to $S$.

\begin{definition}[$p$-popular Spaces]
    Let $V$ be a set where for each $v \in V$ there is an affine space $\calH_v$.
    An affine space $\calH$ is $p$-popular with respect to $V$ if $\Pr{v}{\calH_v = \calH} \ge p$.
\end{definition}

There are at most $1/\eps$ many $\eps$-popular spaces, all of which can be found in $O(\card{S}/\eps)$ time via, say, the Heavy Hitters algorithm of Misra-Gries.

\begin{proposition}[Misra-Gries Heavy Hitters Algorithm \cite{misra1982finding}]
\label{prop:misra-gries}
    The $k$ most frequent elements of a list of length $n$ can be computed in $O(nk)$ time using $O(nk)$ comparisons.
\end{proposition}

When $\dim(\calH) = 1$, this finds all $O(1/\eps)$ codewords from $\calH \cap \calL(g, \Delta-\eps)$ in linear time.

\paragraph{Warm-Up 2: Pruning in the Affine Plane. }
Next, we go over pruning in the case $\dim(\calH) = 2$, which contains the key insight to our algorithm.
Assume for simplicity $m = 1$, and that all field operations and comparisons cost one unit of time.
Let $\calH = \set{h_0 + \alpha_1 h_1 + \alpha_2 h_2: \alpha_1, \alpha_2 \in \F_q}$.

Each coordinate $i \in [N]$ now has associated to it the space $\calH_i$ defined by the linear equality $(h_0 + \alpha_1 h_1 +  \alpha_2 h_2)(i) = g(i)$ where $\alpha_1$ and $\alpha_2$ are the variables.
On the set of indices $S = \set{i: h_1(i) \neq 0 \textrm{ or } h_2(i) \neq 0}$, these spaces are $1$ dimensional.
For any codeword $h \in \calH \cap \calL(g, \Delta-\eps)$, let $S_h$ be the set of coordinates in $S$ for which $h$ agrees with $g$.
Then, $\calH_i$ for all $i \in S_h$ are lines that contain $h$.
By taking a pair of these lines, say $\calH_u,\calH_v$ for $u,v \in S_h$, we either have $\calH_u \cap \calH_v = \set{h}$ or $\calH_u \cap \calH_v = \calH_u$ (when the lines are equal).
Intuitively, if the latter case occurs frequently, then there must be a popular line containing $h$; otherwise, we can pseudorandomly recover $h$ by taking the intersection of two distinct lines --- a win-win scenario.

Concretely, we recursively call the line decoder on all $\Omega(\eps)$-popular lines with respect to $S$, which takes $O(N/\eps^2)$ time total.
If $h$ is in one of these popular lines, we have now recovered it here.
Then, let $G$ be a (constant) $d$-regular expander with vertex set $S$ and to all edges $(u,v) \in E(G)$ we associate the space $\calH_{u,v} = \calH_u \cap \calH_v$.
Since $\card{S} \ge \Delta N$, we have $\card{S_h} \ge \eps N$.
In the case where $h$ is not contained in any popular line with respect to $S$, this would mean that the $0$ dimensional space $\set{h}$ is $\Omega(\eps^2)$-popular with respect to $E(G)$; we will later in the general case formally use \Cref{lem:expander-mixing} (Expander Mixing Lemma) and \Cref{obs:event-partition} to make this claim.
Therefore, inspecting all $\Omega(\eps^2)$-popular points with respect to $E(G)$, each in $O(N)$ time, recovers $h$.
Since there are $O(1/\eps^2)$ such popular points, this sums up to $O(N/\eps^2)$ time.

\subsection{Pruning in Higher Dimensional Spaces}
Having covered the main ideas underlying our result, we now turn to the pruning algorithm in its full generality: $\dim(\calH) = k$.
To each coordinate $i \in [N]$, we associate a space and recursively call the $\le k-1$ dimensional decoder on all $\Omega(\eps)$-popular spaces.
We then generalize the approach for decoding on a line by using a hierarchy of expanders $G_{k-2}, G_{k-3}, \ldots, G_0$ where the vertex set of $G_r$ is defined as the edge set of $G_{r+1}$ (i.e. $V(G_r) = E(G_{r+1})$) and the ``bottom'' graph $G_{k-2}$ plays the role of the expander for line decoding, defined on vertex set $[N]$.
Each edge is associated with the intersection of the spaces associated with its endpoints.
The $\le r$ dimensional decoder is then called on all $\Omega(\eps^{k-r})$-popular spaces with respect to $E(G_r)$ for all $r \in \set{k-2,k-3,\ldots,0}$.
The formal details are described in Algorithm~\ref{alg-2}, with the missing low-level details (of how affine subspaces can be efficiently computed and compared) deferred to later.

\begin{tbox}
    \algname{\textsc{DetPrune} -- Pruning for a Distance $\Delta$ $\F_q$-Linear Code $\calC \subseteq (\F_q^m)^N$}
    \label{alg-2}
    \underline{Input}: Received word $g \in (\F_q^m)^N$.\\
        \phantom{Input:} A fixed constant $\eps \in (0,\Delta)$.\\
        \phantom{Input:} An affine subspace $\calH \subseteq \calC$ over $\F_q$ where $\dim(\calH) = k$.\\
    \underline{Output}: $\calL$, a list of all codewords in $\calH \cap \calL(g, \Delta-\eps)$.
    \begin{enumerate}
        \item If $k = 0$ then halt, outputting $\set{f: f \in \calH \cap \calL(g, \Delta-\eps)}$.
        \item $\calL \gets \emptyset$.
        \item For each $i \in [N]$ compute the affine subspace $\calH_i = \set{h \in \calH: h(i) = g(i)}$.
        \item $V_{k-1} \gets \set{i \in [N]: \dim(\calH_i) \le k-1}$.
        \item $p_{k-1} \gets \eps$.
        \item For each $p_{k-1}/2$-popular space $\underset{\textcolor{gray}{\dim \le k-1}}{\calH'}$ with respect to $V_{k-1}$:
            $\calL \gets \calL \cup \textsc{DetPrune}(g, \eps, \calH').$
        \item For $r$ ranging from $k-2$ down to $0$:
            \begin{enumerate}
                \item $G_{r} \gets$ a $(\card{V_{r+1}}, d = \textrm{poly}(\lambda^{-1}), \lambda = p_{r+1}/24)$-expander on vertex set $V_{r+1}$.
                \item For all $(u,v) \in E(G_{r})$: $\calH_{u,v} \gets \calH_u \cap \calH_v$.
                \item $V_{r} \gets \set{(u,v) \in E(G_{r}): \dim(\calH_{u,v}) \le r}$.
                \item $p_{r} \gets p_{r+1}^2/32$.
                \item For each $p_{r}/2$-popular space $\underset{\textcolor{gray}{\dim \le r}}{\calH'}$ with respect to $V_{r}$:
                    $\calL \gets \calL \cup \textsc{DetPrune}(g, \eps, \calH').$
            \end{enumerate}
        \item \label{alg-2:out} Output $\calL$.
    \end{enumerate}
\end{tbox}

The correctness of Algorithm~\ref{alg-2} is shown by the following lemma, which generalizes (and formalizes) the win-win analysis of the plane decoder to a $\underbrace{\textrm{win-(win-(win-}\ldots\textrm{))}}_{k \textrm{ times}}$ analysis.

\begin{lemma}
\label{lem:alg-2-correct}
    Suppose $h \in \calH \cap \calL(g, \Delta-\eps)$.
    Then $h \in \calL$ on Line~\ref{alg-2:out} of \textsc{DetPrune} (Algorithm~\ref{alg-2}).
\end{lemma}
\begin{proof}
    We will first show by (reverse) induction on $r$ that if $h$ is not added to $\calL$ in any iteration $\ge r$, then $\Pr{v \sim V_{r-1}}{h \in \calH_v} \ge p_{r-1}$.

    As a base case, note that $\Pr{v \sim V_{k-1}}{h \in \calH_v} \ge \eps = p_{k-1}$.
    For any $r \in [k-1]$, suppose $h$ is not added to $\calL$ in iterations $\ge r$.
    We will show that $\Pr{v \sim V_{r-1}}{h \in \calH_v} \ge p_{r-1}$.

    Let $\mathsf{H} = \set{\calH_v : v \in V_{r} \textrm{ and } h \in \calH_v}$ be the set of all affine subspaces containing $h$ that could be obtained by conditioning on vertices in $V_r$. 
    For each $\calH' \in \mathsf{H}$, define the event $E_{\calH'} = \set{ v \in V_r ~\vert~ \calH_v = \calH'}$. Note that the events $\{E_{\calH'}\}_{\calH' \in \mathsf{H}}$ are disjoint.
    
    Let $p \defeq \Pr{v \sim V_{r}}{\cup_{\calH' \in \mathsf{H}} E_{\calH'}}$, and note that by the induction hypothesis, $p \ge p_{r}$.
    Also observe that there is no $\calH' \in \mathsf{H}$ that is $p/2$-popular with respect to $V_{r}$. This is because, by virtue of $h$ not being added to $\calL$ in iteration $r$, there is no $\calH' \in \mathsf{H}$ that is $p_{r}/2$-popular with respect to $V_{r}$ and $p \ge p_{r}$.
    Therefore, we can use \Cref{obs:event-partition} to extract $\mathsf{H}_1 \sqcup \mathsf{H}_2 = \mathsf{H}$ satisfying 
    \[
        p/4 \le \Pr{v \sim V_{r}}{\cup_{\calH' \in \mathsf{H}_i} E_{\calH'}} = \Ex{v\sim V_r}{\mathds{1}_{\mathsf{H}_i}(\calH_v)} \le 3p/4, \quad \text{for }i\in \set{1,2}.
    \]
    Then, we compute,
    {\allowdisplaybreaks
    \begin{align*}
        \Pr{v \sim V_{r-1}}{h \in \calH_v}
        &
        \ge
        \Pr{(u,v) \sim E(G_{r-1})}{h \in \calH_{u,v} \textrm{ and } \dim(\calH_{u,v}) \le r-1}
        \\
        &
        \ge
        \Ex{(u,v) \sim E(G_{r-1})}{\mathds{1}_{\mathsf{H}_1}(\calH_u)\mathds{1}_{\mathsf{H}_2}(\calH_v)}
        \\
        &
        \ge
        \Ex{(u,v) \sim V_{r}^2}{\mathds{1}_{\mathsf{H}_1}(\calH_u)\mathds{1}_{\mathsf{H}_2}(\calH_v)} - \lambda \sqrt{\Ex{u \in V_{r}}{\mathds{1}_{\mathsf{H}_1}^2(\calH_u)}}\sqrt{\Ex{v \in V_{r}}{\mathds{1}_{\mathsf{H}_2}^2(\calH_v)}}
        \tag{Expander Mixing Lemma (\Cref{lem:expander-mixing})}
        \\
        &
        \ge
        \frac{p^2}{16} - \lambda \frac{3p}{4}
        \tag{$p/4 \le \Pr{v \sim V_{r}}{\cup_{\calH' \in \mathsf{H}_i} E_{\calH'}} \le 3p/4$ for $i\in\set{1,2}$}
        \\
        &
        =
        \frac{p^2}{16} - \frac{p\cdot p_r}{32} \tag{$\lambda = p_r/24$} \\
        &
        \ge
        \frac{p_r^2}{16} - \frac{p_r \cdot p_r}{32}
        \tag{\ding{60}}
        \\
        & = \frac{p_{r}^2}{32} = p_{r-1}.
    \end{align*}
    }
    The inequality (\ding{60}) follows from $p \geq p_r$ and the fact that $\frac{x^2}{16} - \frac{x\cdot p_r}{32}$, as a polynomial in $x$, is increasing in the interval $[p_r,p]$.
    This finishes the induction step.

    Now suppose $h$ does not belong to any $p_r/2$-popular subspace for $r > 0$ (and thus has not been added to $\calL$ in iterations $r > 0$), then we can conclude using the proof above that $\Pr{v\in V_0}{h\in \calH_v} \geq p_0$.
    Since there is a unique $0$-dimensional space containing $h$, namely $\set{h}$, it must be that $\set{h}$ is $p_0$-popular with respect to $V_0$.
    Consequently, $h$ is added to $\calL$ when $r = 0$.
\end{proof}


Algorithm~\ref{alg-2} can be implemented to run in near-linear time, as shown by the following proposition.

\begin{proposition}
\label{prop:alg-2-runtime}
    \textsc{DetPrune} (Algorithm~\ref{alg-2}) can be implemented to use at most $(2/\eps)^{O\paren{2^k}} N m \cdot \mathrm{polylog}(q)$ time and $(2/\eps)^{O\paren{2^k}} Nm$ field operations.
\end{proposition}
\begin{proof}
    For simplicity, let us first assume that field operations and comparisons cost one unit of time.
    We will later correct for this by multiplying our time bound under such an assumption by $\textrm{polylog}(q)$.

    \paragraph{Implementation.}
    Below, we show how affine subspaces are represented, which is a crucial detail when invoking \Cref{prop:misra-gries}.
    The input space $\calH$ will be represented by a basis $h_0, h_1, \ldots, h_k$ such that $\calH = \set{h_0 + \sum_{r \in [k]}\alpha_r h_r: \vec{\alpha} \in \F_q^k}$.
    Then, each subspace of $\calH$ will be represented by systems of linear equalities in RREF.
    For example, the initial subspaces $\calH_i$ for $i \in [N]$ will be represented by the following system:
    \begin{align*}
        \begin{pmatrix}
            \spike{20pt}{$h_0(i)$} & \spike{20pt}{$h_1(i)$} & \dots & \spike{20pt}{$h_k(i)$}
        \end{pmatrix}
        \begin{pmatrix}
            1\\
            \alpha_1 \\
            \alpha_2 \\
            \vdots \\
            \alpha_k
        \end{pmatrix}
        =
         \begin{pmatrix}
            \spike{20pt}{$g(i)$}
        \end{pmatrix}
    \end{align*}
    -- where the $h_r(i)$'s and $g(i)$ are columns of length $m$ -- which is then turned into an equivalent system with $O(k)$ rows in RREF using $O(mk^2)$ time.
    The intersections $\calH_u \cap \calH_v$ can thereafter be represented by concatenating the systems of linear equalities for $\calH_u$ with $\calH_v$, and then turning it into RREF in $O(k^3)$ time.

    Once we have the representation of $\calH' \subset \calH$ in terms of a system in RREF, we can compare subspaces in time proportional to the size of their representation, which is $O(k^2)$.
    Finally, if we pass in $\calH' \subset \calH$ to a recursive call, we can compute a basis for $\calH'$ by using $h_0, h_1, \ldots, h_k$ and the RREF system associated with $\calH'$ in time $O(k^3)$.

    \paragraph{Running Time.}
    Let $T(N,k)$ upper bound the time and field operations taken by \textsc{DetPrune} on a block length $N$ input and $k$ dimensional space.
    Clearly, $T(N,0) = O(Nm)$.
    We proceed now by induction on $k$.

    The running time is dominated by the main loop running over $r$ where $0 \le r < k-1$.
    Observe that $p_r = (\eps/2)^{O(2^{k-r})}$ and, consequently, $d_r = (2/\eps)^{O(2^{k-r})}$ where $d_r$ is the degree of $G_r$.
    Therefore,
    \begin{align*}
        &
        \card{E(G_r)} \le \card{E(G_0)} \le \underbrace{\paren{N \prod_{0 \le r < k-1} d_r}}_{\ge \card{V_0}} d_0 = N(2/\eps)^{O\paren{2^k}}.
        \tag{\ding{99}}
    \end{align*}
    The cost $T(N,k)$ is thus bounded above by
    \begin{align*}
        k\paren{
        \underbrace{\log(1/p_0)\card{E(G_0)}}_{\substack{\textrm{Expander construction}\\ \mathrm{\Cref{obs:gg-power}}}}
        +
        \underbrace{(2/p_0)\textrm{poly}(k)m\card{E(G_0)}}_{\substack{\textrm{Find popular subspaces}\\ \mathrm{\Cref{prop:misra-gries}}}}
        +
        \underbrace{(2/p_0) T(N,k-1)}_{\textrm{Recursive calls}}
        }
    \end{align*}
    \begin{sloppypar}
            which, using the induction hypothesis, inequality (\ding{99}), and $1/p_0 \le (2/\eps)^{O(2^{k})}$, is in turn bounded above by $(2/\eps)^{O(2^k)} Nm$.
    \end{sloppypar}
\end{proof}

Put together, we have the following main result.
\begin{theorem}
\label{thm:deterministic-dec}
    Let $\calC \subseteq (\F_q^m)^N$ be a $\Delta$ distance $\F_q$-linear code.
    There is a deterministic algorithm \textsc{DetPrune} that takes in a received word $g \in (\F_q^m)^N$, a fixed constant $\eps \in (0,\Delta)$, and an affine subspace $\calH \subseteq \calC$ over $\F_q$ where $\dim(\calH) = k$ and outputs in $(2/\eps)^{O\paren{2^k}} N m \cdot \mathrm{polylog}(q)$ time a list of all codewords in $\calH \cap \calL(g, \Delta-\eps)$.
\end{theorem}
\begin{proof}
    This immediately follows from \Cref{lem:alg-2-correct} and \Cref{prop:alg-2-runtime}.\qedhere
\end{proof}

\subsection{Deterministic Near-linear Time List Decoding for FRS Codes}
Finally, we show our main application of list decoding FRS codes up to capacity.

\begin{corollary}
\label{cor:deterministic-frs}
Let \(n\) be any positive integer, and let \(\mathbb{F}_q\) be a field with \(q > n\). For all $\eps > 0$ and \(R \in (0, 1)\), there exist \(m\) and an infinite family of \(m\)-folded Reed-Solomon codes \(\mathcal{C}^{\textnormal{FRS}}\) of rate \(R\) and block length \(N = n/m\) that are list-decodable up to radius \(1 - R - \varepsilon\) with a deterministic algorithm in time $(2/\eps)^{O\left(2^{1/\eps}\right)} n \cdot \polylog(n)$.
%
\end{corollary}
\begin{proof}
    We choose $s = 1/\eps$ and $m = s^2$.
    Then, for any received word $g \in (\F_q^m)^N$, we use \Cref{thm:fast-affine-subspace-basis} to extract in $n \cdot \textrm{poly}(1/\eps,\log n)$ time an $O(1/\eps)$ dimensional subspace $\calH$ which contains $\calL(g, 1-R-\eps)$.
    We then pass $g, \eps, \calH$ to \textsc{DetPrune} which, by \Cref{thm:deterministic-dec}, outputs all codewords in $\calL(g, 1-R-\eps)$ in time $(2/\eps)^{O\left(2^{1/\eps}\right)} n \cdot \textrm{polylog}(n)$.
\end{proof}

By using Theorem 2.1 of \cite{GHKS24} in a similar way to \Cref{thm:fast-affine-subspace-basis}, we can also list decode univariate multiplicity codes up to capacity in near-linear time.
\section{A Fully Polynomial-Time Pruning Algorithm}

In this section, we present a {fully polynomial-time} randomized algorithm that list decodes Folded Reed-Solomon (FRS) codes up to a radius of \(1 - R - \varepsilon\) for any \(\varepsilon > 0\). At the core of the algorithm is a subroutine called \textsc{RandPrune} (Algorithm~\ref{alg:prune}), which takes as input a received word \(g \in (\mathbb{F}_q^m)^N\) and an affine subspace \(\mathcal{H}\), and with some probability, returns an FRS codeword in \(\mathcal{H}\) that is close to \(g\). To recover the complete list of codewords near \(g\), one simply repeats this subroutine sufficiently many times on an affine subspace that contains the list.

At a high level, given a received word \(g\) and an affine subspace \(\mathcal{H}\), the \textsc{RandPrune} subroutine first computes, for all coordinates \(i \in [N]\), the affine subspace
\[
\mathcal{H}_i = \{ h \in \mathcal{H} : h \text{ agrees with } g \text{ at coordinate } i \}.
\]
Then it recursively calls \textsc{RandPrune} on \(\mathcal{H}_i\) for a randomly chosen $i\in [N]$, continuing until it reaches a subspace of dimension zero, at which point it returns the unique codeword in that subspace. The key difference in our algorithm compared to the KRSW algorithm is that the coordinate $i$ is not selected uniformly at random; instead, we sample a subspace with probability roughly proportional to its dimension. The following subsection goes into more detail.
\subsection{More Efficient Pruning}
Recall from \cref{sec:our-results} that, for any affine subspace \(\mathcal{H} \subseteq \frs\) of dimension \(k\), and any received word \(g \in (\mathbb{F}_q^m)^N\), the KRSW algorithm recovers each codeword \(h \in \mathcal{L}(g, 1 - R - \varepsilon)\) with probability at least \(\varepsilon^k\). A key insight in that analysis was that, for a uniformly chosen coordinate \(i \in [N]\), the probability that the dimension of \(\mathcal{H}_i\) drops by at least 1 compared to \(\mathcal{H}\) is lower bounded by \(1 - R\). In the case of Folded Reed-Solomon codes, however, the dimension typically drops by much more than 1. Specifically, it can be shown (\cref{lemma: GK16}) that
\[
\ExpOp_{i \sim [N]}[\dim(\mathcal{H}_i)] \leq kR(1 + \Theta(\varepsilon)).
\]
This structural fact was exploited in~\cite{Sri25} to improve list size bounds for Folded Reed-Solomon codes,\footnote{As mentioned before, Algorithm~\ref{alg:prune} may be viewed as an algorithmization of the list size bound proof in~\cite{Sri25}.} and it is precisely what enables us to replace the weak inductive hypothesis \(\varepsilon^k\) with the much stronger \(\varepsilon / k\).

If we want to make the inductive step work with this new, stronger hypothesis, however, we can no longer sample coordinates uniformly at random. Nevertheless, we show that selecting a coordinate \(i\) with probability roughly proportional to \(\dim(\mathcal{H}_i)\) is enough to complete the induction step for the stronger hypothesis. The complete \textsc{RandPrune} subroutine, omitting low-level details regarding subspace representation during execution, is provided in Algorithm~\ref{alg:prune}.

\begin{tbox}
    \algname{\textsc{RandPrune} -- Pruning for a Rate $R$ FRS Code $\frs \sub (\F_q^m)^N$} 
    \label{alg:prune}
    \underline{Input}: A received word $g \in (\F_q^m)^N$.\\
        \phantom{Input:} An affine subspace $\mathcal{H} \subseteq \mathcal{C}^{\textnormal{FRS}}$ of dimension $k$.\\
        \phantom{Input:} An integer \(s\) satisfying \(k\leq s\leq m\).\\
    \underline{Output}: Some $h \in \mathcal{L}\left(g, \frac{s}{s+1}\left(1 - \frac{mR}{m-s+1}\right)\right) \cap \mathcal{H}$ or FAIL.
    \begin{enumerate}
        \item \label{alg:rand-base} If $k = 0$ then halt, outputting the unique codeword \(h\in \mathcal{H}\) if \(\Delta(g, h) < \frac{s}{s+1}\left(1 - \frac{mR}{m - s + 1}\right)\), else FAIL.
        \item For each $i \in [N]$ compute the affine subspace $\calH_i = \set{h \in \calH: h(i) = g(i)}$.
        \item For each $r \in \set{0,1,\ldots,k}$:
            \begin{align*}
                &S_r \leftarrow \{i \in [N] : \mathcal{H}_i \neq \emptyset \text{ and } \dim(\mathcal{H}_i) = r\}
                \\
                &w_r \leftarrow
                \begin{cases}
                    |S_r|\cdot (sr + 1), & \textnormal{if } r\neq k\\
                    0, & \textnormal{if } r=k
                \end{cases}
            \end{align*}
        \item $W \leftarrow \sum_{r=0}^{k} w_r$, halting and outputting FAIL if $W = 0$.
        \item Sample $r^* \in \{0, 1, \dots, k\}$ with probability $\frac{w_{r^*}}{W}$.
        \item Sample $i$ uniformly at random from $S_{r^*}$.
        \item \label{alg:recurse-line} Output \textsc{RandPrune}$(g, \mathcal{H}_i, s)$
    \end{enumerate}
\end{tbox}











The key result of this section is the following claim.

\begin{theorem}
\label{thm:randomized-success}
    Let \(g\in (\mathbb{F}_q^m)^N\) be a received word, \(\mathcal{H}\subseteq \mathcal{C}^{\textnormal{FRS}}\) be an affine subspace of dimension \(k\) and \(s\geq k\). For every \(h\in \mathcal{L}\left(g, \frac{s}{s+1}(1 - \frac{mR}{m-s+1})\right) \cap \mathcal{H}\), the probability that \textsc{RandPrune}\((g, \mathcal{H}, s)\) returns \(h\) is at least \(\frac{1}{sk+1}\).
\end{theorem}
\begin{proof}
    The proof is by induction on the dimension \(k\) of the affine subspace \(\mathcal{H}\). Clearly, the statement holds when \(k=0\) since in that case there is exactly one codeword \(h\in \mathcal{H}\). Furthermore, if \(\Delta(g, h) < \frac{s}{s+1}(1-\frac{mR}{m-s+1})\), then the algorithm outputs \(h\) with probability 1 on Line~\ref{alg:rand-base}. So, let \(k\geq 1\) and assume that the claim holds for all affine subspaces of dimension less than \(k\). Fix some \(h\in \mathcal{L}\left(g, \frac{s}{s+1}(1 - \frac{mR}{m-s+1})\right) \cap \mathcal{H}\), and let \(A\subseteq [N]\) be the set of coordinates where \(g\) agrees with \(h\), i.e.,
    \[A:= \{i\in [N]: g(i)=h(i)\}.\]
    Note that we have \(|A| \geq \left(1-\frac{s}{s+1}(1-\frac{mR}{m-s+1})\right)N\). Furthermore, for each \(r\in \{0, 1, \ldots , k\}\), define
    \(A_r = A\cap S_r\). We can now write ---
    \allowdisplaybreaks
    \begin{align*}
    \prob[\textsc{RandPrune}(g, \mathcal{H}, s) \text{ outputs } h]
    &= \frac{1}{W} \sum_{r=0}^{k} w_r  \Pr{i \sim S_r}{i \in A_r \land \textsc{RandPrune}(g, \mathcal{H}_i, s) \text{ outputs } h} \\
    & = \frac{1}{W} \sum_{r=0}^{k-1} w_r  \prob_{i \sim S_r}\left[i \in A_r\right] \cdot \prob_{i \sim S_r}\left[\textsc{RandPrune}(g, \mathcal{H}_i, s) \text{ outputs } h ~\vert~ i\in A_r \right] \\
    &\geq \frac{1}{W} \sum_{r=0}^{k-1} |S_r| \cdot (sr + 1) \cdot \frac{|A_r|}{|S_r|} \cdot \frac{1}{sr + 1} \\
    &= \frac{1}{W} \sum_{r=0}^{k-1} |A_r| \\
    &= \frac{1}{W} \left(|A| - |A_k|\right) \\
    & \geq \frac{1}{W} \left(|A| - |S_k|\right) \\
    &\geq \frac{1}{W}\left(\left(1 - \frac{s}{s+1}\left(1-\frac{mR}{m-s+1}\right)\right)N - |S_k|\right) \\
    &= \frac{1}{W}\left(\left(\frac{1}{s+1} - \frac{s}{s+1}\cdot\frac{mR}{m-s+1}\right)N - |S_k|\right).
    \end{align*}
    Next, we will upper bound the value of \(W\). To do so, we will use the following lemma due to Guruswami and Kopparty~\cite{GK16}. We refer the reader to \cite[Theorem~5.2]{Sri25} or \cite[Lemma~3.1]{GHKSS25} for self-contained proofs.
    \begin{lemma}[\cite{GK16}]
    \label{lemma: GK16}
        \( \sum_{i=1}^N \dim(\calH_i) = \sum_{r=0}^k r|S_r| \leq k\cdot \frac{mR}{m-k+1}\cdot N\).
    \end{lemma}
    We can now write ---
    \allowdisplaybreaks
    \begin{align*}
W &= \sum_{r=0}^{k} w_r = \sum_{r=0}^{k-1} w_r \\
  &= \sum_{r=0}^{k-1} |S_r| \cdot (sr + 1) \\
  &= s\sum_{r=0}^{k-1} r|S_r| + \sum_{r=0}^{k-1} |S_r| \\
  &\leq s \left(k \cdot \frac{mR}{m-k+1} \cdot N - k |S_k|\right) + \sum_{r=0}^{k-1} |S_r| & [\textnormal{using~\cref{lemma: GK16}}]\\
  &\leq s k\left(\frac{mR}{m-k+1} \cdot N - |S_k|\right) + (N-|S_k|) & [\textnormal{using the disjointness of the \(S_r\)'s}]\\
  &\leq sk \left(\frac{mR}{m-s+1} N - |S_k|\right) + (N - |S_k|). & [\textnormal{since \(s\geq k\)}]\\
  &= \parens*{1+sk\cdot \frac{mR}{m-s+1}}N - (sk+1)|S_k|
\end{align*}
Therefore, we have --- \[\prob[\textsc{RandPrune}(g, \mathcal{H}, s) \text{ outputs } h] \geq \frac{\left(\frac{1}{s+1} - \frac{s}{s+1}\cdot\frac{mR}{m-s+1}\right)N - |S_k|}{\parens*{1+sk\cdot \frac{mR}{m-s+1}}N - (sk+1)|S_k|},\]
and we would like to show that
\begin{equation*}
    \frac{\left(\frac{1}{s+1} - \frac{s}{s+1}\cdot\frac{mR}{m-s+1}\right)N - |S_k|}{\parens*{1+sk\cdot \frac{mR}{m-s+1}}N - (sk+1)|S_k|} \geq \frac{1}{sk+1}.
\end{equation*}
This is equivalent to showing
\begin{equation*}
    (sk+1)\left(\left(\frac{1}{s+1} - \frac{s}{s+1}\cdot\frac{mR}{m-s+1}\right)N - |S_k|\right) - \left(\parens*{1+sk\cdot \frac{mR}{m-s+1}}N - (sk+1)|S_k|\right) \geq 0.
\end{equation*}
The term involving $|S_k|$ cancels, and after dividing by $N$, we are left with
\begin{align*}
	&(sk+1)\left(\frac{1}{s+1} - \frac{s}{s+1}\cdot\frac{mR}{m-s+1}\right)  - \parens*{1+sk\cdot \frac{mR}{m-s+1}} \\
	=~& (sk+1)\parens*{\frac{1}{s+1} - \frac{1}{s+1} \cdot \frac{mR}{m-s+1}} - 1+ \frac{mR}{m-s+1} \\
	=~& \frac{sk+1}{s+1}\parens*{1-\frac{mR}{m-s+1}} - \parens*{1-\frac{mR}{m-s+1}} \\
	=~& \parens*{\frac{sk+1}{s+1}-1}\parens*{1-\frac{mR}{m-s+1}} \\
	=~& \frac{s(k-1)}{s+1}\parens*{1-\frac{mR}{m-s+1}} \geq 0. &&[\text{using }k\geq 1]
\end{align*}\qedhere
\end{proof}

We now proceed to analyze the running time of Algorithm~\ref{alg:prune}. Note that, as in~\cref{sec:deterministic}, subspaces are represented using a basis in this algorithm as well.

\begin{proposition}
\label{prop:prune-runtime}
Let \(g \in (\mathbb{F}_q^m)^N\) be a received word, \(\mathcal{H} \subseteq C^{\textnormal{FRS}}\) be an affine subspace of dimension \(k\), and \(s \geq k\). Then \textsc{RandPrune}\((g, \mathcal{H}, s)\) performs at most \(O(mk^3 N)\) field operations over \(\mathbb{F}_q\).
\end{proposition}

\begin{proof}
The recursive call to \textsc{RandPrune} on Line~\ref{alg:recurse-line} is made on a subspace \(\mathcal{H}_i\) with strictly smaller dimension, i.e., \(\dim(\mathcal{H}_i) < k\). Hence, the depth of recursion is at most \(k\). For each coordinate \(i \in [N]\), checking whether \(\mathcal{H}_i\) is empty or computing a basis for \(\mathcal{H}_i\) requires solving a system of \(m\) linear constraints on at most \(k\) variables over \(\mathbb{F}_q\), which can be done using Gaussian elimination in \(O(mk^2)\) field operations. Since this check is performed across all \(N\) coordinates at each recursion level, the total number of field operations is bounded by \(O(mk^2 N)\) per level. Multiplying by the worst-case recursion depth \(k\) yields an overall count of \(O(mk^3 N)\) field operations.
\end{proof}

\subsection{Fully Polynomial Time List Decoding for FRS Codes}
We now use Algorithm~\ref{alg:prune} to give a list decoding algorithm for Folded Reed-Solomon codes. Recall~\cref{thm:list-in-affine-subspace}, which states that the list of close codewords for Folded Reed-Solomon codes lies in a low-dimensional affine subspace \(\mathcal{H}\). In what follows, let \(T\) denote the number of field operations required to compute a basis for this subspace \(\mathcal{H}\). We now state the main list decoding result.

\begin{theorem}
\label{thm:repeat-randomized}
Let \(\mathcal{C}^{\textnormal{FRS}}\) be an \(m\)-Folded Reed-Solomon code of block length \(N = n/m\) and rate \(R\). Let \(g \in (\mathbb{F}_q^m)^N\) be a received word, and let \(s \in [m]\) be any integer. Then for any \(\beta > 0\), there exists a randomized algorithm that performs at most \(T+O\left(ms^5 \log\left(\tfrac{s}{\beta}\right) N\right)\) field operations over \(\mathbb{F}_q\) and outputs
\(
\mathcal{L}\left(g, \frac{s}{s+1} \left(1 - \frac{mR}{m - s + 1}\right)\right)
\)
with probability at least \(1 - \beta\).
\end{theorem}

\begin{proof}
\begin{sloppypar}
    Given \(g\) and \(s\), we first compute a basis for the subspace \(\mathcal{H}\) from~\cref{thm:list-in-affine-subspace}. We then run \textsc{RandPrune}\((g, \mathcal{H}, s)\) independently \(t := \left\lceil(s(s - 1) + 1) \ln\left(\tfrac{s}{\beta}\right)\right\rceil\) times, and return the union \(\mathcal{L}^*\) of all the codewords found.
\end{sloppypar}
From~\cref{thm:FRS-gen-singleton=bound}, we know that the list size satisfies \(\left| \mathcal{L}\left(g, \frac{s}{s+1} \left(1 - \frac{mR}{m - s + 1}\right)\right) \right| \leq s\). Therefore, by~\cref{thm:randomized-success} and a union bound over the \(s\) possible codewords, we have
\allowdisplaybreaks
\begin{align*}
\prob\left[\exists h \in \mathcal{L}\left(g, \frac{s}{s+1} \left(1 - \frac{mR}{m - s + 1}\right)\right) \setminus \mathcal{L}^* \right]
\leq s \left(1 - \frac{1}{s(s - 1) + 1}\right)^t
\leq s \cdot e^{-t / (s(s - 1) + 1)} 
\leq \beta.
\end{align*}
\begin{sloppypar}
   Finally, by \cref{prop:prune-runtime}, each call to \textsc{RandPrune} performs at most \(O(ms^3 N)\) field operations over \(\mathbb{F}_q\), and so the total field operation count is \(T+O\left(ms^5 \log\left(\tfrac{s}{\beta}\right) N\right)\). 
\end{sloppypar}
\end{proof}
Finally, we can use \cref{thm:repeat-randomized} to list decode Folded Reed-Solomon codes up to capacity.

\begin{corollary}\label{cor:randomized_frs}
Let \(n\) be any positive integer, and let \(\mathbb{F}_q\) be a field with \(q > n\). For all \(\varepsilon, \beta > 0\) and \(R \in (0, 1)\), there exist \(m\) and an infinite family of \(m\)-Folded Reed-Solomon codes \(\mathcal{C}^{\textnormal{FRS}}\) of rate \(R\) and block length \(N = n/m\) that are list-decodable up to radius \(1 - R - \varepsilon\) in time \(\poly\left(\frac{1}{\varepsilon}, \log\left(\frac{1}{\beta}\right), \log q\right) \cdot \widetilde{O}(N)\), with success probability at least \(1 - \beta\).
\end{corollary}

\begin{proof}
We instantiate \cref{thm:repeat-randomized} with the parameters \(s \geq \frac{1}{\varepsilon}\) and \(m\geq s^2\). It is not hard to verify that with these settings,
\(
\frac{s}{s+1} \left(1 - \frac{mR}{m - s + 1} \right) \geq 1 - R - \varepsilon.
\)
Hence, the resulting code \(\mathcal{C}^{\textnormal{FRS}}\) is algorithmically list-decodable up to radius \(1 - R - \varepsilon\), with success probability at least \(1 - \beta\), using
\(T+
\text{poly}\left(\tfrac{1}{\varepsilon}, \log\left(\tfrac{1}{\beta}\right)\right)\cdot O(N)
\)
field operations over \(\mathbb{F}_q\). By \cref{thm:fast-affine-subspace-basis}, we have \(T =  \text{poly}(\tfrac{1}{\varepsilon}) \cdot \widetilde{O}(N)\). Therefore, the total number of field operations is bounded by \(\poly\left(\frac{1}{\varepsilon}, \log\left(\frac{1}{\beta}\right)\right) \cdot \widetilde{O}(N)\). Since each field operation takes \(\polylog(q)\) time, the claim follows.
\end{proof}

\section*{Acknowledgements}
We thank Rishabh Kothary for pointing out an error in an older version of \Cref{obs:event-partition}.

\bibliographystyle{alphaurl}
\bibliography{macros,references}

\end{document}